\documentclass[a4paper, 12pt]{amsart}
\usepackage{amsmath, amssymb, eucal, amscd, amstext, enumerate, mathrsfs, yfonts, amsfonts}

\newtheorem{thm}{Theorem}[section]
\newtheorem{lem}[thm]{Lemma}

\newtheorem{prop}[thm]{Proposition}
\newtheorem{cor}[thm]{Corollary}

\theoremstyle{definition}

\theoremstyle{remark}

\newtheorem{rem}[thm]{Remark}

\newtheorem{eg}[thm]{Example}
\newtheorem{question}[thm]{Question}

\newenvironment{prf}{{\noindent \textbf{Proof:}\ }}{\hfill $\Box$\\ \smallskip}

\numberwithin{equation}{section}

\newcommand{\sa}{{\rm sa}}

\newcommand{\ti}{\tilde}
\newcommand{\la}{\langle}
\newcommand{\ra}{\rangle}

\newcommand{\RP}{\mathbb{R}_+}

\newcommand{\BC}{\mathbb{C}}

\newcommand{\BN}{\mathbb{N}}
\newcommand{\BR}{\mathbb{R}}
\newcommand{\BT}{\mathbb{T}}

\newcommand{\CC}{\mathcal{C}}
\newcommand{\CF}{\mathcal{F}}

\newcommand{\CL}{\mathcal{L}}

\newcommand{\CI}{\mathcal{I}}

\newcommand{\KI}{\mathfrak{I}}

\newcommand{\KS}{\mathfrak{S}}

\newcommand{\KH}{\mathfrak{H}}
\newcommand{\KP}{\mathfrak{P}}

\newcommand{\ms}{\mathbf{s}}
\newcommand{\me}{\mathrm{e}}

\newcommand{\dist}{\mathrm{d}}
\newcommand{\R}{\mathbf{R}}
\newcommand{\bt}{\mathbf{t}}

\newcommand{\Proj}{\mathcal{P}}
\newcommand{\Cent}{\mathcal{Z}}

\begin{document}

\title[Transition probabilities preservers]{Transition probabilities of normal states determine the Jordan structure of a quantum system}

\author{Chi-Wai Leung \and Chi-Keung Ng \and Ngai-Ching Wong}

\address[Chi-Wai Leung]{Department of Mathematics, The Chinese
	University of Hong Kong, Hong Kong.}
\email{cwleung@math.cuhk.edu.hk}

\address[Chi-Keung Ng]{Chern Institute of Mathematics and LPMC, Nankai University, Tianjin 300071, China.}
\email{ckng@nankai.edu.cn}

\address[Ngai-Ching Wong]{Department of Applied Mathematics, National Sun Yat-sen University,  Kaohsiung, 80424, Taiwan.}
\email{wong@math.nsysu.edu.tw}


\date{\today}



\begin{abstract}
Let $\Phi:\mathfrak{S}(M_1)\to \mathfrak{S}(M_2)$ be a bijection (not assumed  affine nor continuous) between the sets of normal states of two quantum systems, modelled on the self-adjoint parts of  von Neumann algebras $M_1$ and $M_2$, respectively.
This paper concerns with the situation when $\Phi$ preserves (or partially preserves) one of the following three notions of ``transition probability'' on the normal state spaces: the Uhlmann transition probability $P_U$,
the Raggio transition probability $P_B$ and an ``asymmetric transition probability'' $P_0$ as defined in this article.

It is shown that the two systems are isomorphic, i.e.\
$M_1$ and $M_2$ are Jordan $^*$-isomorphic,
if $\Phi$ preserves all pairs with zero Uhlmann (respectively, Raggio or asymmetric) transition probability, i.e., for any normal states $\mu$ and $\nu$, we have
$$
P\big(\Phi(\mu),\Phi(\nu)\big) = 0 \quad \text{if and only if} \quad P(\mu,\nu)=0,
$$
where $P$ stands for $P_U$ (respectively, $P_R$ or $P_0$).
Furthermore, as an extension of Wigner's theorem, it is shown that there is a Jordan $^*$-isomorphism $\Theta:M_2\to M_1$ with
$$\Phi = \Theta^*|_{\mathfrak{S}(M_1)}$$
if and only if $\Phi$ preserves the ``asymmetric transition probability''.
This is also equivalent to $\Phi$ preserving the Raggio transition probability.
Consequently, if $\Phi$ preserves the Raggio transition probability, it will  preserve the Uhlmann transition probability as well.
As another application, the sets of normal states equipped with either the usual metric, the Bures metric or ``the metric induced by the self-dual cone'' are complete Jordan $^*$-invariants for the underlying von Neumann algebras.
\end{abstract}

\maketitle

\section{Introduction}


Let  $\KH_1$ and $\KH_2$ be two (complex) Hilbert spaces and $T:\KH_1\to \KH_2$ be a bijective map (not assumed  linear nor continuous).
Wigner's theorem states that if  $T$  preserves the transition probability, in the sense that
$$|\la T(\xi), T(\eta)\ra|^2 = |\la \xi, \eta\ra |^2 \qquad (\xi,\eta\in \KH_1),$$
then there exist  a unitary or an anti-unitary $S:\KH_1\to \KH_2$ and a function $f:\KH_2\to \BT$ such that $T(\xi) = f(\xi)S(\xi)$ ($\xi\in \KH_1$).
 Uhlhorn's theorem, as a generalization of Wigner's theorem, states that if $\dim \KH_1\geq 3$
and $T$  preserves pairs with zero transition probability,
in the sense that
$$\la T(\xi), T(\eta)\ra = 0 \quad \text{if and only if} \quad \la \xi, \eta\ra =0 \qquad (\xi,\eta\in \KH_1),$$
then there exist  a unitary or an anti-unitary $S$ and a function $g:\KH_2\to \BC\setminus \{0\}$ such that $T(\xi) = g(\xi)S(\xi)$ ($\xi\in \KH_1$).


Let $A$ be a (complex) $C^*$-algebra and $\mu, \nu\in A^*$ be pure states of $A$.
The \emph{transition probability} between $\mu$ and $\nu$ is defined to be the quantity
$$P(\mu, \nu):=\mu(\ms_\nu),$$
where $\ms_\nu$ is the support projection of $\nu$ in $A^{**}$.
It is well-known that $P(\mu,\nu) = P(\nu,\mu)$, i.e., $\mu(\ms_\nu) = \nu(\ms_\mu)$, for pure states $\mu$ and $\nu$ (see e.g.\ \cite{AS03}).
Suppose that $\pi:A\to \CL(\KH)$ is a $^*$-representation of $A$ and $\xi\in \KH$, we denote, as usual,
\begin{equation}\label{eqt:def-omega-xi}
\omega_\xi(x):= \la \pi(x)\xi, \xi\ra \qquad (x\in A).
\end{equation}
In the case when $A=\CL(\KH)$ and $\pi:\CL(\KH)\to \CL(\KH)$ is the default representation, the functionals $\omega_\xi$ and $\omega_\eta$ are pure
normal states of $\CL(\KH)$ (where $\xi,\eta\in \KH$) and we have
\begin{equation*}
P(\omega_\xi, \omega_\eta) = |\la \xi, \eta\ra|^2.
\end{equation*}
In this setting, Wigner's (respectively, Uhlhorn's) theorem can be interpreted as structural results concerning bijections between pure normal state spaces of $\CL(\KH_1)$ and $\CL(\KH_2)$ that preserve (respectively, partially preserve) the transition probability, and several proofs are given (see e.g. \cite{Geh14} or  \cite[Theorem 1]{Shultz98}).
They have also been extended to the setting of indefinite inner product spaces by Moln\'{a}r (see \cite[Theorem 1]{Mol00} and \cite[Corollary 1]{Mol02}).
Through our study we will also give another proof for Wigner's theorem (see Corollary \ref{cor:Wigner}).


On the other hand, Shultz provided a throughout study of transition probability preserving bijections
between pure state spaces of general $C^*$-algebras.
Under some extra conditions, such maps are induced by
the dual maps of algebraic or Jordan $^*$-isomorphisms of the $C^*$-algebras (see e.g.,  \cite{AS01, AS03, Shultz98} for details).
Related considerations of maps between pure state spaces of $C^*$-algebras preserving transition probability or other properties
can also be found in, e.g., \cite{Brown92, Lab,  Stormer63}.


However, the pure state setting of transition probability is inappropriate to be adapted to the case of von Neumann algebras.
Unlike  $\CL(\KH)$,
a general von Neumann algebra may not have any pure normal state at all.
Therefore, people are looking for suitable notions of transition probability on the
space $\KS(M)$ of all normal states on a von Neumann algebra $M$ (see e.g. \cite{AU00, BS14, Rag84, Schmud13, Uhlmann11}).
Here, by a normal state on $M$, we means a norm one positive normal linear functional on $M$, and it is different from the notion of ``physical states'' as introduced in \cite{Christen-82}.


Let  $\R(M)$ denote the collection of all (unitary equivalence classes of) faithful unital $^*$-representations of
a von Neumann algebra $M$.
For any $\mu, \nu\in \KS(M)$ and $(\KH, \pi)\in \R(M)$, we set $\KH(\mu):= \{\xi\in \KH:  \omega_\xi = \mu\}$ (could be empty).
The quantity
$$P_U(\mu, \nu):= \sup\{ |\la \xi, \eta\ra|^2: \xi\in \KH(\mu), \eta\in \KH(\nu), (\KH, \pi)\in \R(M)\}$$
is well-defined and is called the \emph{Uhlmann transition probability} of $\mu$ and $\nu$ (\cite{Uhlmann76}).
The Uhlmann transition probability is related to the so-called \emph{Bures distance} $\dist_B$ through the formula
\begin{equation}\label{eqt:def-d-B}
\dist_B(\mu, \nu) := \sqrt{2 - 2 \sqrt{P_U(\mu, \nu)}}.
\end{equation}
This metric $\dist_B$ is in general different from the usual distance $\dist_1$ on $\KS(M)$ given by
$$
\dist_1(\mu, \nu):=\|\mu - \nu\|.
$$


In \cite{Rag}, Raggio defined another transition probability.
Suppose that $(M, \KH, \KP, J)$ is the standard form for $M$ as in \cite{Haag-st-form} (see Section 2 below for a brief exploration).
By \cite[Lemma 2.10]{Haag-st-form}, for any $\mu\in \KS(M)$, there is a unique $\xi_\mu\in \KP$ satisfying
\begin{equation}\label{eqt:nor-st-P}
\mu = \omega_{\xi_\mu}.
\end{equation}
If $\mu, \nu\in \KS(M)$, the positive real number
$$P_R(\mu, \nu):= \la \xi_\mu,\xi_\nu\ra$$	
is called the \emph{Raggio transition probability} of $\mu$ and $\nu$.
As in the Uhlmann case, the Raggio transition probability induces a metric on $\KS(M)$ by
\begin{equation}\label{eqt:Rag-tran-prob<->d2}
\dist_2(\mu, \nu) := \sqrt{2 - 2 \sqrt{P_R(\mu, \nu)}} \qquad (\mu,\nu\in \KS(M)).
\end{equation}
This metric coincides with the one induced from $\KH$, namely,
\begin{equation}\label{eqt:d2<->Hil}
\dist_2(\mu,\nu) = \|\xi_\mu - \xi_\nu\| \qquad (\mu, \nu\in \KS(M)).
\end{equation}
In \cite[Corollary 1]{Rag}, the following relation between the Raggio and the Uhlmann transition probabilities was presented:
\begin{equation}\label{eqt:cf-Rag-Uhlm}
P_U(\mu, \nu) \leq P_R(\mu, \nu) \leq P_U(\mu, \nu)^{1/2} \qquad (\mu, \nu\in \KS(M)).
\end{equation}


In addition, there is a more na\"{i}ve extension of the ``transition probability'':
\begin{equation}\label{eqt:defn-asym-tran-prob}
P_0(\mu, \nu) := \mu(\ms_\nu) \qquad (\mu,\nu\in \KS(M)).
\end{equation}
\emph{Strictly speaking, $P_0$ is not a transition probability,} because unlike the two extensions above, $P_0$ is asymmetric, and $P_0(\mu,\nu)=1$ is equivalent to $\ms_\mu\leq \ms_\nu$ instead of $\mu=\nu$ (c.f. \cite[p.325]{Rag84}).
Nevertheless, abusing the language, we still call $P_0$ the ``\emph{asymmetric transition probability}''.
It seems to be conceptual clear and technically easier to work with it.


Notice that two normal states $\mu,\nu\in \KS(M)$ are orthogonal, i.e., having orthogonal support projections, exactly when
they have zero transition probability in any (and equivalently, all) of the above three settings (see \eqref{eqt:orth<->van-tran-prob} in Section 3).


The main concern of this article is on those bijections (not assumed  affine nor continuous) from the normal state space of one von Neumann algebra to that of another preserving either one of the three transition probabilities above.
We obtain
two analogues of Wigner's theorem for bijections between normal state spaces of general quantum systems (which are modelled on self-adjoint elements of von Neumann algebras).
Furthermore, several weak analogues of Uhlhorn's theorem for normal state spaces of general quantum systems were also obtained.

More precisely, it is shown that a bijection between normal state spaces preserving either the ``asymmetric transition probability'' (as defined in \eqref{eqt:defn-asym-tran-prob}) or the Raggio transition probability was shown to be induced by a Jordan $^*$-isomorphism  (see Theorems \ref{thm:pre-tran-prop>jord-isom}(b) and \ref{thm:Raggio-tran-prop<->jord-isom}).
The result concerning the ``asymmetric transition probability'' can be regarded as an extension of the original Wigner's theorem because of Corollary \ref{cor:Wigner}.
Moreover, we verified that the normal state space equipped with either the Uhlmann transition probability, the Raggio transition probability or the ``asymmetric transition probability'', completely identifies the underlying quantum system (see Theorems \ref{thm:pre-tran-prop>jord-isom} and \ref{thm:st-sp-met-inv}).
Consequently, bijections between normal state spaces preserving the Raggio transition probability will preserve the Uhlmann transition probability (see Corollary \ref{cor:pre-Raggio>pre-Uhlmann}).

This study highlighted the importance of the Raggio and the Uhlmann transition probability in quantum mechanics and it also established a strong relation between these two notions of transition probability.
On the other hand, the notion of ``asymmetric transition probability'' that defined in \eqref{eqt:defn-asym-tran-prob} seems to be conceptually clearer and easier to implement in physics, although it is not strictly speaking a transition probability.
Actually, Theorem \ref{thm:pre-tran-prop>jord-isom}(b) implies that the datum of measurements of observables associated with support projections of states at all other states is sufficient to determine the  quantum system completely.

In developing our main results, we also obtained that several metric spaces associated with the sets of normal states of von Neumann algebras (without any algebraic structure) are complete Jordan $^*$-invariants for the underlying algebras (see Corollary \ref{cor:comp-Jord-inv}).

\section{Notations and Preliminaries}


Throughout this article, $M$, $M_1$ and $M_2$ are (complex) von Neumann algebras.
We denote by $\KS(M)$ and $\Proj(M)$ the normal state space of $M$  and the set of all projections in $M$, respectively.

Let $(M, \KH, \KP, J)$ denote the (unique) standard form of a von Neumann algebra $M$ (see \cite{Haag-st-form}).
In other words, $\KH$ is a (complex) Hilbert space with $M$ being a (unital) von Neumann subalgebra of $\CL(\KH)$, $J$ is a conjugate linear isometric involution on $\KH$ and $\KP\subseteq \KH$ is a cone which is \emph{self-dual}, in the sense that
$$\KP = \{\eta\in \KH: \la \eta,\xi\ra \geq 0, \text{ for any }\xi\in \KP\},$$
such that the following conditions hold:
\begin{enumerate}
\item $JMJ=M'$,
\item $JcJ = c^*$ for any $c\in \Cent(M)$,
\item $J\xi = \xi$ for any $\xi\in \KP$,
\item $aa^\bt(\KP) \subseteq \KP$ for any $a\in M$,
\end{enumerate}
where $M'$ is the commutant of $M$  in $\CL(\KH)$, $\Cent(M) := M\cap M'$ and $a^\bt:=JaJ$.
We put
$$
S_{\KH}:= \{\xi\in \KH: \|\xi\| = 1\} \quad \text{and} \quad S_\KP := \KP\cap S_\KH.
$$


\begin{rem}\label{rem:WOT=weak*}
Suppose that $\{x_i\}_{i\in \KI}$ is a net in $M$ that WOT-converges to $x\in M$, when considered as operators in $\CL(H)$.
As $\omega_\xi(x_i) \to \omega_\xi(x)$ ($\xi\in \KP$) and $\{\omega_\xi:\xi\in \KP\} = M_*^+$, we know that $\{x_i\}_{i\in \KI}$ weak$^*$-converges to $x$.
Thus, the WOT on $M\subseteq \CL(\KH)$ coincides with the weak$^*$-topology.
In particular, if $\{e_i\}_{i\in \KI}$ is an increasing net in $\Proj(M)$ with $e_i\uparrow e_0\in \Proj(M)$, then $ e_i y e_i \stackrel{w^*}{\longrightarrow} e_0 y e_0$  ($y\in M$),
because $\omega_\xi(e_i y e_i) = \la y e_i \xi, e_i\xi\ra \to \omega_\xi(e_0 y e_0)$ ($\xi\in \KP$).
\end{rem}


The following proposition contains some known results.
Note that Part (b) of it  inherits from \cite[Theorem 2.2]{CKLW}, while part (a) can be regarded as a result of Dye, because all the ingredients for its proof are already in \cite{Dye} (and a similar discussion can be found in \cite{Sher}, although it is not explicitly stated there).
We give a simple proof for part (b) here, so that we have a complete elementary proof for Corollary \ref{cor:Wigner} below.


\begin{prop}\label{prop:Dye}
Let $M_1$ and $M_2$ be two von Neumann algebras.
Suppose that there is an \emph{orthoisomorphism} $\Gamma: \Proj(M_1)\to \Proj(M_2)$, i.e. $\Gamma$ is bijective, and for any $p,q\in \Proj(M_1)$,
$$p\,q=0  \quad \text{if and only if}\quad \Gamma(p) \Gamma(q) = 0. $$

\noindent
(a) (Dye) $M_1$ and $M_2$ are Jordan $^*$-isomorphic.

\noindent
(b) If $\Gamma$ extends to a bounded linear map $\ti \Gamma:M_1^\sa\to M_2^\sa$, then $\ti \Gamma$ is a Jordan isomorphism.
\end{prop}
\begin{prf}
(a) By \cite[Lemma 1]{Dye}, the bijection $\Gamma$ is an order isomorphism that sends central projections to central projections.
Let $\me_k$ be the central projection in $M_k$ such that $\me_k M_k$ is the type $I_2$ part of $M_k$ ($k=1,2$).

We first show that $\Gamma(\me_1)= \me_2$ and $(1-\me_1) M_1$ is Jordan $^*$-isomorphic to $(1-\me_2) M_2$.
In fact, as $\Gamma$ is an order isomorphism with $\Gamma(1-\me_1)\in \Cent(M_2)$, it restricts to an orthoisomorphism  from $\Proj((1-\me_1) M_1)$ onto $\Proj(\Gamma(1-\me_1)M_2)$.
The absence of nonzero type $I_2$ summand in $(1-\me_1)M_1$ and the Corollary in \cite[p.~83]{Dye}
ensure that $\Gamma|_{\Proj((1-\me_1)M_1)}$ extends to a Jordan $^*$-isomorphism from $(1-\me_1) M_1$ onto $\Gamma(1-\me_1)M_2$.
Hence, $\Gamma(1-\me_1)M_2$ does not have a nonzero type $I_2$ summand neither.
This means $\Gamma(1-\me_1)\me_2 = 0$, or equivalently,
$$\Gamma(1-\me_1)\leq 1-\me_2.$$
Similarly, $\Gamma^{-1}(1-\me_2)\leq 1-\me_1$.
By \cite[Lemma 1]{Dye}, one has $1 - \Gamma(\me_1)= \Gamma(1-\me_1) = 1-\me_2$.

It remains to show that $\me_1 M_1$ is Jordan $^*$-isomorphic to $\me_2 M_2$.
Indeed, because $\Gamma(\me_1) =\me_2$, the map $\Gamma$ restricts to an orthoisomorphism from
$\Proj(\me_1M_1)$ onto $\Proj(\me_2M_2)$.
Since
$$\Gamma\big(\Cent(\me_1M_1)\cap \Proj(\me_1M_1)\big) = \Cent(\me_2M_2)\cap \Proj(\me_2M_2),$$
$\Gamma$ induces an orthoisomorphism from $\Proj(\Cent(\me_1M_1))$ onto $\Proj(\Cent(\me_2M_2))$, and the Corollary in \cite[p.~83]{Dye} implies that $\Cent(\me_1M_1)$ is $^*$-isomorphic to $\Cent(\me_2M_2)$.
The conclusion now follows from the fact that $\me_kM_k = \Cent(\me_kM_k)\otimes \mathrm{M}_2(\BC)$ ($k=1,2$).

\noindent
(b) Let $n\in \BN$ and $r_1,...,r_n\in \BR$.
Consider $\{p_1,...,p_n\}$ to be a set of orthogonal elements in $\Proj(M_1)$.
If we denote $a:= \sum_{k=1}^n r_k p_k$, then we have
$$\ti \Gamma(a^2) = \sum_{k=1}^n r_k^2 \Gamma(p_k) = \ti \Gamma\Big(\sum_{k=1}^n r_kp_k\Big)^2$$
(the last equality uses the fact that $\Gamma$ preserves orthogonality).
Since elements with finite spectrum is norm-dense in $M_1^\sa$, the continuity of $\ti \Gamma$ ensures that $\ti \Gamma(x^2) = \ti \Gamma(x)^2$ ($x\in M_1^\sa$), as claimed.
\end{prf}


Our next proposition is also well-known (although we do not find the second conclusion explicitly stated anywhere), and we will only give a brief account for it.


\begin{prop}\label{prop:J-preserves-U-R-TP}
Let $\Theta: M_2\to M_1$ be a Jordan $^*$-isomorphism.
Then $\Theta^*(\KS(M_1)) = \KS(M_2)$ and $\Theta^*|_{\KS(M_1)}$  preserves both the Raggio and the Uhlmann transition probabilities.
\end{prop}
\begin{prf}
By \cite[Theorem 10]{Kad51}, there is a central projection $f_2\in \Proj(M_2)\cap \Cent(M_2)$ such that $\Theta$ restricts to a $^*$-isomorphism $\Theta_i: f_2M_2 \to f_1M_1$ as well as a $^*$-anti-isomorphism $\Theta_a: (1-f_2)M_2 \to (1-f_1)M_1$, where $f_1:=\Theta(f_2)$.
For $k\in \{1,2\}$, one has
$$\KS(M_k) = \left\{\big(t\mu_k, (1-t)\nu_k\big): t\in [0,1]; \mu_k\in \KS(f_kM_k); \nu_k\in \KS\big((1-f_k)M_k\big)\right\}$$
and $\R(M_k) = \R(f_kM_k) \times \R\big((1-f_k)M_k\big)$ in the canonical way.

Since all $^*$-isomorphisms and $^*$-anti-isomorphisms between von Neumann algebras are isometric and weak$^*$-continuous, the same is true for Jordan $^*$-isomorphisms and we have $\Theta^*(\KS(M_1)) = \KS(M_2)$.

Note that $\Theta_i$ and $\Theta_a$ induce
canonical bijections from $\R(f_1M_1)$ onto $\R(f_2M_2)$ and from $\R\big((1-f_1)M_1\big)$ onto $\R\big((1-f_2)M_2\big)$,
respectively.
Indeed,   if $(\pi, \KH)$ is a $^*$-representation of $(1-f_1)M_1$ and $(\bar \pi, \bar \KH)$ is its conjugate representation, then $(\bar \pi \circ \Theta_a, \bar \KH)$ is a $^*$-representation of $(1-f_2)M_2$.
Now, it follows  from the definition of $P_U$ that $\Theta^*|_{\KS(M_1)}$ preserves $P_U$.

On the other hand, if $(M_k, \KH_k, \KP_k, J_k)$ is the standard form of $M_k$ ($k=1,2$), then
the uniqueness of the standard form ensures that $(M_k, \KH_k, \KP_k, J_k)$ decomposes canonically with respect to $f_k$ and $1-f_k$.
Moreover, if $(M, \KH, \KP, J)$ is the standard form of a von Neumann algebra $M$, then $x\mapsto Jx^*J$ is a $^*$-anti-isomorphism from $M$ to $M'$,
and $(M',\KH, \KP, J)$ is the standard form of $M'$.
From these, one can check easily that $\Theta^*|_{\KS(M_1)}$ preserves $P_R$.
\end{prf}


\section{The main results}


Set $\Proj_\sigma(M) := \{\ms_\mu: \mu\in \KS(M)\}$.
For any $p\in \Proj(M)$, it follows from  Zorn's Lemma that there is an orthogonal family $\{p_i\}_{i\in \KI}$ in $\Proj_\sigma(M)$   satisfying
\begin{equation}\label{eqt:proj-decomp}
p=\sum_{i\in \KI} p_i
\end{equation}
(the convergence is taken in the weak$^*$-topology).
We write
\begin{equation*}\label{eqt:defn-F0}
F_0(p) := \{\nu\in \KS(M): \nu(p) = 0\}.
\end{equation*}
Obviously, $F_0(p)$ coincides with the closed face $(1-p)\KS(M)(1-p)\cap \KS(M)$ of $\KS(M)$.
Moreover, since
\begin{align}\label{eqt:orthog-support}
\nu(\ms_\mu) = 0\quad\text{if and only if}\quad \ms_\nu\ms_\mu = 0,
\end{align}
we have
\begin{equation}\label{eqt:rel-F-mu-bot}
F_0(\ms_\mu)\ =\ \{\nu\in \KS(M): \ms_\nu \ms_\mu = 0\} \qquad (\mu\in \KS(M)).
\end{equation}
If $p = \sum_{i\in \KI} p_i$ is as in \eqref{eqt:proj-decomp}, then
\begin{equation}\label{eqt:F0-decomp}
F_0(p) = \bigcap_{i\in \KI} F_0(p_i).
\end{equation}
See, e.g., \cite{Stormer68}, for more explorations between projections and their associated faces.

We say that a map $\Phi: \KS(M_1) \to \KS(M_2)$ is \emph{biorthogonality preserving} if for any $\mu,\nu\in\KS(M_1)$, one has
$$\ms_\mu \ms_\nu = 0 \quad \text{if and only if} \quad \ms_{\Phi(\mu)}\ms_{\Phi(\nu)} = 0.$$


\begin{lem}\label{lem:st-sp-to-proj}
Let $M_1$ and $M_2$ be two von Neumann algebras.
Suppose that $\Phi: \KS(M_1) \to \KS(M_2)$ is a biorthogonality preserving bijection.

\noindent
(a) There exists an orthoisomorphism $\check \Phi:\Proj(M_1)\to \Proj(M_2)$ such that
\begin{equation}\label{eqt:full-rel-Psi-F0}
\Phi(F_0(p)) = F_0(\check \Phi(p)) \quad (p\in \Proj(M_1)) \quad \text{and} \quad \check \Phi(\ms_\mu) = \ms_{\Phi(\mu)} \quad (\mu\in \KS(M_1)).
\end{equation}

\noindent
(b) If $M_1$ or $M_2$ does not have a type $I_2$ summand, then $\check \Phi$ extends to a Jordan isomorphism from $M_1^\sa$ to $M_2^\sa$.

\noindent
(c) If $\Theta:M_2\to M_1$ is a Jordan $^*$-isomorphism satisfying $\Phi = \Theta^*|_{\KS(M_1)}$, then
$$
\Phi(\nu) (\ms_{\Phi(\mu)}) = \nu(\ms_\mu)\quad (\mu, \nu\in \KS(M_1)), \quad\text{and}\quad \check \Phi = \Theta^{-1}|_{\Proj(M_1)}.
$$
\end{lem}
\begin{proof}
(a) We denote by $\CF(M_k)$ the set of all closed faces of $\KS(M_k)$ ($k=1,2$).
The bijectivity of $\Phi$ and \eqref{eqt:rel-F-mu-bot} tell us that $\Phi$ is biorthogonality preserving if and only if
\begin{equation}\label{eqt:Rel-F0-Psi}
\Phi(F_0(\ms_\mu)) = F_0(\ms_{\Phi(\mu)}) \qquad (\mu\in \KS(M_1)).
\end{equation}
Let $p\in \Proj(M_1)$, and $p:= \sum_{i\in \KI} \ms_{\mu_i}$ be a decomposition as in \eqref{eqt:proj-decomp} for a family $\{\mu_i\}_{i\in \KI}$ in $\KS(M_1)$ with its elements having disjoint support projections.
By  the hypothesis, elements in $\{\Phi(\mu_i)\}_{i\in \KI}$ have disjoint support projections, and hence $\sum_{i\in \KI} \ms_{\Phi(\mu_i)}$ converges in the weak$^*$-topology to a projection $\check \Phi(p)\in \Proj(M_2)$.
Since $\Phi$ is injective, \eqref{eqt:F0-decomp} and \eqref{eqt:Rel-F0-Psi} imply
\begin{equation*}
\Phi(F_0(p)) = \bigcap_{i\in \KI} \Phi(F_0(\ms_{\mu_i})) = \bigcap_{i\in \KI} F_0(\ms_{\Phi(\mu_i)}) = F_0\big(\sum_{i\in \KI} \ms_{\Phi(\mu_i)}\big) = F_0(\check \Phi(p)).
\end{equation*}
Moreover, the map $F_0:p\mapsto F_0(p)$ is a bijection from $\Proj(M_k)$ onto $\CF(M_k)$ for $k=1,2$ (\cite[Theorem 3.35]{AS03}).
These show that $\check \Phi(p)$ is independent of the choice of $\{\mu_i\}_{i\in \KI}$, and that $\Phi$ induces a map $\Phi^\CF: \CF(M_1) \to \CF(M_2)$.

In the same way,  $\Phi^{-1}$ induces a map from $\CF(M_2)$ to $\CF(M_1)$ which is clearly the inverse of $\Phi^\CF$.
Therefore, $\Phi^\CF$ is a bijection, and the bijectivity of the map $\check\Phi:\Proj(M_1)\to \Proj(M_2)$ follows from the bijectivity of $F_0$.

Suppose now that $p,q\in \Proj(M_1)$ satisfying $pq = 0$.
Then for any $p',q'\in \Proj_\sigma(M_1)$ with $p'\leq p$ and $q'\leq q$, one has $p'q'= 0$.
Hence, from the hypothesis concerning $\Phi$ and the definition of $\check \Phi$, we conclude  that $\check \Phi(p)\check \Phi(q) = 0$.
Again, by considering $\Phi^{-1}$, we know that $\check\Phi$ is an orthoisomorphism.

\noindent
(b) This follows from the Corollary in \cite[p.~83]{Dye}.

\noindent
(c) By second equality in \eqref{eqt:full-rel-Psi-F0},
$$
\mu(\Theta(\check{\Phi}(\ms_\mu))) = \mu(\Theta(\ms_{\Phi(\mu)})) = \Theta^*(\mu)(\ms_{\Theta^*(\mu)}))=1.
$$
Thus, $\ms_\mu\leq \Theta(\check{\Phi}(\ms_\mu)) = \Theta (\ms_{\Theta^*(\mu)})$.
Conversely, as $\mu(\ms_\mu) =1$, one has $\Theta^*(\mu)(\Theta^{-1}(\ms_\mu)) = 1$, which means that $\ms_{\Theta^*(\mu)} \leq \Theta^{-1}(\ms_\mu)$.
These give
\begin{equation}\label{eqt:ms-mu-Theta}
\ms_\mu = \Theta(\check{\Phi}(\ms_\mu)) = \Theta (\ms_{\Theta^*(\mu)}),
\end{equation}
and hence $\Phi(\nu) (\ms_{\Phi(\mu)}) = \nu(\ms_\mu)$.
On the other hand, due to the construction of $\check \Phi$ in the argument for part (a), Equality \eqref{eqt:ms-mu-Theta} also produces the second conclusion.
\end{proof}




\begin{thm}\label{thm:pre-tran-prop>jord-isom}
Let $M_1$ and $M_2$ be von Neumann algebras, and let $\Phi:\KS(M_1)\to \KS(M_2)$ be a bijection.

\noindent
(a) If $\Phi$ is   biorthogonality preserving, then $M_1$ and $M_2$ are Jordan $^*$-isomorphic.

\noindent
(b) There is a Jordan $^*$-isomorphism $\Theta:M_2\to M_1$ satisfying
$\Phi = \Theta_*|_{\KS(M_1)}$ if and only if $\Phi$ preserves the ``asymmetric transition probability'' $P_0$, i.e.,
$$P_0\big(\Phi(\mu),\Phi(\nu)\big) = P_0(\mu, \nu)  \qquad (\mu, \nu\in \KS(M_1)).$$
\end{thm}
\begin{prf}
(a)  This follows directly from Lemma \ref{lem:st-sp-to-proj}(a) and Proposition \ref{prop:Dye}(a).

\noindent
(b) Suppose that such a Jordan $^*$-isomorphism $\Theta$ exists.
Then Lemma \ref{lem:st-sp-to-proj}(c) tells us that $P_0\big(\Phi(\nu),\Phi(\mu)\big) = P_0(\nu, \mu)$ ($\mu, \nu\in \KS(M_1)$).

For the converse implication, we first note that because $\mu(\ms_\nu) = \Phi(\mu)(\ms_{\Phi(\nu)})$ ($\mu, \nu\in \KS(M)$), the map $\Phi$ is biorthogonality preserving (see \eqref{eqt:orthog-support}).
Consider $\check \Phi: \Proj(M_1)\to \Proj(M_2)$ to be the map as in Lemma \ref{lem:st-sp-to-proj}.
Let $M_1^0$ be the real linear span of $\Proj(M_1)$ in $M_1$.
We want to extend $\check \Phi$ to $M_1^0$ by setting
$$
\check \Phi(x) := \sum_{k=1}^n r_k \check \Phi(p_k),
$$
when $x=\sum_{k=1}^n r_k p_k$ for some $n\in \BN$, $r_1,...,r_n\in \BR$ and $p_1,...,p_n\in \Proj(M_1)$.
To show that this extension is well-defined, let us consider $p_k = \sum_{i \in \KI_k} \ms_{\mu_{k,i}}$ to be a decomposition as in \eqref{eqt:proj-decomp} ($k=1,...,n$).
By the construction of $\check \Phi$ in the proof of Lemma \ref{lem:st-sp-to-proj}(a), for any $\mu\in \KS(M_1)$,
$$\Phi(\mu)\Big(\sum_{k=1}^n r_k \check\Phi(p_k)\Big)
\ =\ \sum_{k=1}^n r_k \sum_{i\in \KI_k} \Phi(\mu)\big(\ms_{\Phi(\mu_{k,i})}\big)
\ =\ \sum_{k=1}^n r_k \sum_{i\in \KI_k} \mu\big(\ms_{\mu_{k,i}}\big)
\ =\ \mu(x).$$
Thus, the surjectivity of $\Phi$ implies that $\check \Phi(x)$ is independent of the choices of $r_1,...,r_n$ nor $p_1,...,p_n$.
Obviously, $\check\Phi$ is a linear map on $M_1^0$ satisfying
\begin{equation}\label{eqt:ext-ti-Phi}
\nu\big(\check \Phi(x)\big) = \Phi^{-1}(\nu)(x)  \qquad (x\in M_1^0; \nu\in \KS(M_2)).
\end{equation}
This implies $\|\check \Phi(x)\| = \|x\|$ $(x\in M_1^0)$, and $\check \Phi$ extends to a linear isometry from $M_1^\sa$ onto $M_2^\sa$.
By Proposition \ref{prop:Dye}(b), this extension, and hence its inverse $\Theta:M_2^\sa\to M_1^\sa$ is a Jordan isomorphism.
Furthermore, Relation \eqref{eqt:ext-ti-Phi} tells us that $\Phi = \Theta_*|_{\KS(M_2)}$.
\end{prf}

Note that one may use, for example, \cite[Theorem 2]{Kad52} instead of Proposition \ref{prop:Dye}(b) to conclude the proof of part (b) above.
However, the current proof enables us to get an elementary proof of our next corollary  (note that the results in \cite{Dye} are only needed in Theorem \ref{thm:pre-tran-prop>jord-isom}(a)).

Through the argument of this corollary, one may also regard Theorem \ref{thm:pre-tran-prop>jord-isom}(b) as an extension of Wigner's theorem, and hence we get another proof for Wigner's theorem.
Notice that the first part of the proof of this corollary is similar to that of \cite[Theorem 1]{Shultz98}, but instead of showing the extension to be affine and employing \cite[Corollary 5]{Kad52}, we show that the extension preserves $P_0$ and use Theorem \ref{thm:pre-tran-prop>jord-isom}(b) to obtain the conclusion.


\begin{cor}\label{cor:Wigner} (Wigner)
Let $\KH_1$ and $\KH_2$ be two Hilbert spaces and let
$$\KS_\mathrm{p}(\CL(\KH_k)) := \{\omega_\xi: \xi\in S_{\KH_k}\} \qquad (k=1,2).$$
If $\Phi: \KS_\mathrm{p}(\CL(\KH_1))\to  \KS_\mathrm{p}(\CL(\KH_2))$ is a bijection that preserves the transition probability, there is a Jordan $^*$-isomorphism $\Theta:\CL(\KH_2)\to \CL(\KH_1)$ with
$\Phi = \Theta_*|_{\KS_\mathrm{p}(\CL(\KH_1))}$.
\end{cor}
\begin{prf}
For any $k=1,2$ and $\xi\in S_{\KH_k}$, we know that $\ms_{\omega_\xi}$ is the projection $e_\xi$ from $\KH_k$ onto $\BC\cdot \xi$.
Through diagonalisation of positive trace-class operators, we know that for each $\mu\in \KS(\CL(\KH_1))$, there exist $n\in \BN\cup\{\infty\}$, an orthonormal sequence $\{\xi_i\}_{i=1}^n$ in $S_{\KH_1}$ and  a sequence $\{t_i\}_{i=1}^n$ in $(0,1]$ with $\sum_{i=1}^n t_i = 1$ such that
$\mu =\sum_{i=1}^n t_i \omega_{\xi_i}$ (converges in norm).
In this case, we propose to set
$$
\bar \Phi(\mu) := \sum_{i=1}^n t_i  \Phi(\omega_{\xi_i}).
$$
For any finite orthonormal sequence $\{\zeta_j\}_{j=1}^N$ in $S_{\KH_1}$, one has, by the hypothesis,
\begin{equation}\label{eqt:ext-Phi}
\sum_{i=1}^n t_i  \Phi(\omega_{\xi_i})\Big(\sum_{j=1}^N\ms_{\Phi(\omega_{\zeta_j})}\Big)
\ = \ \sum_{j=1}^N\sum_{i=1}^n t_i  \omega_{\xi_i}(\ms_{\omega_{\zeta_j}})
\ = \ \mu\Big(\sum_{j=1}^N\ms_{\omega_{\zeta_j}}\Big).
\end{equation}
Since $\Phi$ is surjective, the above tells us  that the value of $\sum_{i=1}^n t_i  \Phi(\omega_{\xi_i})$ on any finite rank projection in $\Proj(\CL(\KH_2))$ is independent of the decomposition $\mu =\sum_{i=1}^n t_i \omega_{\xi_i}$.
Thus, $\bar \Phi(\mu)$ is well-defined.

On the other hand, Relation \eqref{eqt:ext-Phi} also  tells us that $\bar \Phi: \KS(\CL(\KH_1)) \to \KS(\CL(\KH_2))$ is an injection, and the surjectivity of $\bar \Phi$ follows from
the surjectivity of $\Phi$.
Furthermore, \eqref{eqt:ext-Phi} implies that $\bar\Phi(\mu)(\ms_{\Phi(\omega_\eta)}) = \mu(\ms_{\omega_\eta})$ ($\eta\in S_{\KH_1}$).

Let $\nu\in \KS(\CL(\KH_1))$ and $\nu = \sum_{l=1}^m r_l \omega_{\eta_l}$ be a
decomposition of $\nu$ similar to that of $\mu$ in the above.
As $\{\ms_{\omega_{\eta_l}}\}_{l=1}^m$ (and hence $\{\ms_{\Phi(\omega_{\eta_l})}\}_{l=1}^m$) is an orthogonal sequence and $r_l> 0$ (for all $l=1,2,...,m$), we have
\begin{equation*}\label{eqt:decomp-s-Phi-mu}
\ms_{\nu} = \sum_{l=1}^m \ms_{\omega_{\eta_l}} \quad \text{as well as}\quad \ms_{\bar\Phi(\nu)} = \sum_{l=1}^m \ms_{\Phi(\omega_{\eta_l})}
\end{equation*}
(the convergences are in the weak$^*$-topology).
Hence,
$$\bar \Phi(\mu)(\ms_{\bar\Phi(\nu)})
\ =\ \bar \Phi (\mu) \Big( \sum_{l=1}^m \ms_{\Phi(\omega_{\eta_l})}\Big)
\ =\ \sum_{l=1}^m \mu(\ms_{\omega_{\eta_l}})
\ =\ \mu(\ms_{\nu}).$$
Finally, Theorem \ref{thm:pre-tran-prop>jord-isom}(b) gives a Jordan $^*$-isomorphism  $\Theta:\CL(\KH_1)\to \CL(\KH_2)$ satisfying $\Phi=\Theta_*\mid_{\KS(\CL(\KH_1))}$.
\end{prf}


On the other hand, Theorem \ref{thm:pre-tran-prop>jord-isom}(a) can  be regarded as an extension of a weak form of Uhlhorn's theorem for the normal state space of von Neumann algebras.
In particular, we have the following application of it.


\begin{thm}\label{thm:st-sp-met-inv}
Let $M_1$ and $M_2$ be von Neumann algebras.
Then $M_1$ and $M_2$ are Jordan $^*$-isomorphic if there is a bijection $\Phi:\KS(M_1) \to \KS(M_2)$ satisfying any
one of the following conditions:
\begin{enumerate}[(1)]
\item $\Phi$ preserves the usual metric $\dist_1$;
\item $\Phi$ preserves pairs with zero Raggio transition probabilities, i.e.\
for any $\mu,\nu\in \KS(M_1)$,
$$P_R\big(\Phi(\mu), \Phi(\nu)\big)=0 \quad \text{if and only if} \quad P_R(\mu, \nu)=0;$$
\item $\Phi$ preserves pairs with zero Uhlmann transition probabilities;
\item $\Phi$ preserves pairs with zero ``asymmetric transition probabilities''.
\end{enumerate}
\end{thm}
\begin{prf}
We claim that in each of the four cases, $\Phi$ is biorthogonality preserving, and thus Theorem \ref{thm:pre-tran-prop>jord-isom}(a) applies.

Indeed, the assertion for the case of
$\Phi$ preserving the usual metric $\dist_1$ follows from the  well-known fact that $\ms_\mu\ms_\nu = 0$ if and only if $\|\mu - \nu\| = 2$.

Suppose that $\Phi$ preserves pairs with zero Raggio transition probabilities.
By \eqref{eqt:Rag-tran-prob<->d2} and \eqref{eqt:d2<->Hil}, we know that $P_R(\mu, \nu) = 0$ if and only if $\|\xi_\mu - \xi_\nu\|^2 = 2$.
On the other hand, it follows from \cite[Lemma 2.10(2)]{Haag-st-form} that $\|\xi_\mu - \xi_\nu\|^2 = 2$ if and only if $\|\mu - \nu\| = 2$, because $\|\xi_\mu - \xi_\nu\|\|\xi_\mu + \xi_\nu\| = \sqrt{4 - 4\la \xi_\mu, \xi_\nu\ra^2}$.
Thus, the assertion for the second case follows from that of the first case.

Finally, the assertions for the third (respectively, the fourth) case, follows from \eqref{eqt:cf-Rag-Uhlm} (respectively, \eqref{eqt:orthog-support}) and the second case.
\end{prf}


As seen in the above, for any $\mu,\nu\in \KS(M_1)$, one has
\begin{equation}\label{eqt:orth<->van-tran-prob}
\ms_\mu\ms_\nu = 0 \quad \Longleftrightarrow\quad P_U(\mu, \nu) = 0 \quad  \Longleftrightarrow\quad  P_R(\mu, \nu) = 0 \quad  \Longleftrightarrow\quad  P_0(\mu, \nu) = 0.
\end{equation}
In particular, we obtained an alternative proof of \cite[Lemma 1.8]{AP00}.


One may wonder if it is possible to get a stronger conclusion for Theorem \ref{thm:pre-tran-prop>jord-isom}(a) (and hence a stronger conclusion for Theorem \ref{thm:st-sp-met-inv}) similar to that of Theorem \ref{thm:pre-tran-prop>jord-isom}(b).
However, the following example shows that it is impossible even in the case when  $M_1 = M_2 = \CL(\ell^2)$.


\begin{eg}\label{eg:Uhlh-gen-fails}
	Let $M= \CL(\KH)$ with $\operatorname{dim} \KH \geq 2$.   Let $\sim$ be an equivalence relation in $\KS(M)$ defined by
	$$
	\mu \sim \nu \quad\text{if and only if}\quad \ms_\mu=\ms_\nu.
	$$
	Let $\CC$ be the set of equivalence classes of $\KS(M)$ under $\sim$.
	Suppose that $\zeta_1$ and $\zeta_2$ are two orthogonal elements in $S_{\KH}$, and $e_{\zeta_k}\in \Proj(M)$ is the orthogonal projection onto $\BC\cdot \zeta_k$ ($k=1,2$).
	For any $t\in (0,1)$, if we set $\mu_t := t \omega_{\zeta_1} + (1-t) \omega_{\zeta_2}$, then $\ms_{\mu_t} = e_{\zeta_1} + e_{\zeta_2}$.
	Hence, $\{\mu_t: t\in (0,1)\}\subseteq C_0$ for an element $C_0\in \CC$.
	Consider any bijection $\Phi_0:C_0\to C_0$ with
	$$\Phi_0(\mu_t) = (1-t) \omega_{\zeta_1} + t \omega_{\zeta_2} \qquad (t\in (0,1)),$$
	and define a bijection $\Phi:\KS(M)\to \KS(M)$ by setting $\Phi|_{C_0} = \Phi_0$ as well as
	$$
	\Phi(\mu) = \mu \qquad (\mu\in \KS(M)\setminus C_0).
	$$
	From the definition of $\Phi$, we know that $\ms_{\Phi(\mu)} = \ms_\mu$ ($\mu\in \KS(M)$), and 
	$\Phi$ is biorthogonality preserving.
	However, since
	$$\|\omega_{\zeta_1} - \mu_t\| = 2-2t \quad \text{and} \quad \|\Phi(\omega_{\zeta_1}) - \Phi(\mu_t)\| = 2t \qquad (t\in (0,1)),$$
	one concludes that $\Phi$ cannot be induced by any continuous map from $M_*$ to itself.
\end{eg}


Nevertheless, in the case when the bijection $\Phi$ actually preserves the Raggio transition probability, we will see
in Theorem \ref{thm:Raggio-tran-prop<->jord-isom} below that the conclusion as in Theorem \ref{thm:pre-tran-prop>jord-isom}(b) holds.
In order to obtain this result, we need some more preparations.


A normed space $X$ is said to be \emph{strictly convex} if for any $x,y\in X$, the condition $\|x+y\| = \|x\|+\|y\|$ implies that $x$ and $y$ are linearly dependent.
Clearly, any Hilbert space is strictly convex.
Let us recall the following well-known fact in Banach spaces theory.


\begin{lem}\label{lem:strict-conv-pos-homo}
	Suppose that $X_1$ and $X_2$ are real Banach spaces such that $X_2$ is strictly convex.
	If $K$ is a convex subset of $X_1$ and $f:K\to X_2$ is a metric preserving map, then $f$ is automatically an affine map.
\end{lem}


For the benefit of the readers, we sketch a proof here.
In fact, in order to show
$$f\big(t x + (1-t)y\big) = tf(x) + (1-t)f(y)$$
for any $x\neq y$ in $K$ and $t\in (0,1)$, we may assume
(by ``shifting'' $K$ and $f$ if necessary) that $y=0$ and that $f(0) = 0$.
In this case, we have
\begin{equation}\label{eqt:affine}
\|f(x) - f(tx)\| = \|x -tx\| = (1-t)\|f(x)\| = \|f(x)\| - t\|x\| = \|f(x)\| - \|f(tx)\|,
\end{equation}
and the strict convexity gives $f(x) - f(tx)\in \BR\cdot f(tx)$.
This, together with \eqref{eqt:affine},  establishes  the required relation:
$f(tx) = tf(x)$.


\begin{prop}\label{prop:isom-cone>Jord}
Let $(M_k, \KH_k, J_k, \KP_k)$ be a von Neumann algebra in its standard form ($k=1,2$).
There are canonical bijective correspondences (through restrictions) amongst the following:
\begin{itemize}
		\item the set $\CI_\KH$ of complex linear isometries from $\KH_1$ to $\KH_2$ sending $\KP_1$ onto $\KP_2$;
		\item the set $\CI_\KP$ of metric preserving surjections from $\KP_1$ onto $\KP_2$;
		\item the set $\CI_S$ of metric preserving surjections from $S_{\KP_1}$ onto $S_{\KP_2}$.
\end{itemize}
\end{prop}
\begin{prf}
For every $\rho\in \CI_\KH$, one clearly has $\rho|_{S_{\KP_1}}\in \CI_S$.
The assignment $\rho\mapsto \rho|_{S_{\KP_1}}$ defines an injection $R:\CI_\KH \to \CI_S$,  because $S_{\KP_1}$ generates $\KH_1$.
Secondly, if $\chi\in \CI_S$, then
$$\la \chi(\xi), \chi(\eta)\ra = \la \xi, \eta \ra\in \RP \qquad (\xi,\eta\in S_{\KP_1}),$$
and hence the extension $\ti \chi:t\xi \mapsto t\chi(\xi)$ ($\xi\in S_{\KP_1}, t\in \RP$) belongs to $\CI_\KP$.
This gives an injection $E:\CI_S \to \CI_\KP$.
Furthermore, as elements in $\CI_\KH$ are affine, the composition $E\circ R: \CI_\KH\to \CI_\KP$ coincides with the restriction map $\rho\mapsto \rho|_{\KP_1}$.
Thus, it remains to show that $E\circ R$ is surjective.
	
	Let us now consider $\varphi\in \CI_\KP$.
	Since the only extreme point in $\KP_k$ is the zero element, we know from Lemma \ref{lem:strict-conv-pos-homo} that $\varphi(0) = 0$.
The metric preserving assumption now implies
	\begin{equation}\label{eqt:pre-met}
	\|\varphi(\xi)\| = \|\xi\| \quad \text{and}\quad \la \varphi(\xi),\varphi(\eta)\ra = \la \xi, \eta\ra\in \RP \qquad (\xi,\eta\in \KP_1).
	\end{equation}
For $k=1,2$, we denote by $\KH_k^\sa$ the real Hilbert space $\KP_k - \KP_k$.
As $\KP_k$ is a self-dual cone, if $\eta \in \KH_k^\sa$, there exist unique elements $\xi^+, \xi^-\in \KP_k$ with $\xi = \xi^+ - \xi^-$ and $\|\xi\|^2 =  \|\xi^+\|^2 +  \|\xi^-\|^2$ (see e.g.\ \cite[Lemme I.1.2]{Ioch-sd-cone}).

Define $\ti\varphi: \KH_1^\sa\to \KH_2^\sa$ by
	$$\ti \varphi(\xi) := \varphi(\xi^+) - \varphi(\xi^-) \qquad (\xi\in \KH_1^\sa).$$
	For every $\xi,\eta\in \KH_1^\sa$, one knows from \eqref{eqt:pre-met} that
	\begin{equation*}
		\|\ti \varphi(\xi) - \ti \varphi(\eta)\|^2
		\ = \ \|\varphi(\xi^+)  - \varphi(\xi^-) - \varphi(\eta^+)+\varphi(\eta^-)\|^2
		\ = \ \|\xi^+ -\xi^- - \eta^+ + \eta^- \|^2,
	\end{equation*}
which means that $\ti \varphi$ preserves metric.
Hence,
$$\|\ti \varphi(\xi)\|^2 = \|\xi\|^2 = \|\xi^+\|^2 + \|\xi^-\|^2 = \|\varphi(\xi^+)\|^2 + \|\varphi(\xi^-)\|^2,$$ and the uniqueness of $\ti\varphi(\xi)^\pm$ produces
	\begin{equation}\label{eqt:varphi-pre-pm}
	\ti\varphi(\xi)^\pm = \varphi(\xi^\pm) \qquad (\xi\in \KH_1^\sa).
	\end{equation}

	If $\psi:= \varphi^{-1}:\KP_2\to \KP_1$ and $\ti \psi$ is defined in the same way as $\ti \varphi$, then, by a similar property as \eqref{eqt:varphi-pre-pm} for $\ti \psi$, we obtain that, for each $\zeta\in \KH_1^\sa$,
	\begin{equation*}
		\ti \varphi(\ti \psi(\zeta))
		\ = \ \ti \varphi(\ti\psi(\zeta)^+ - \ti\psi(\zeta)^-)
		\ = \ \varphi(\ti\psi(\zeta)^+) - \varphi(\ti\psi(\zeta)^-)
		\ = \ \varphi(\psi(\zeta^+)) - \varphi(\psi(\zeta^-))
		\ = \ \zeta.
	\end{equation*}
Consequently, $\ti \varphi$ is surjective.
It now follows from the Mazur-Ulam theorem that $\ti \varphi$ is a linear isometry from $\KH^{sa}_1$ onto $\KH^{sa}_2$.
Finally, the complexification, $\bar \varphi$, of $\ti \varphi$ is an element in  $\CI_\KH$ (note that linear isometries preserve inner products) satisfying $\bar\varphi|_{\KP_1} = \varphi$.
\end{prf}


Recall that a projection $p\in \Proj(M)\setminus \{0\}$ is said to be \emph{$\sigma$-finite} if any family of non-zero orthogonal subprojections of $p$ is countable.
It is easy to check that  $\Proj_\sigma(M)$ consists exactly of  $\sigma$-finite projections and the sum of a finite number of orthogonal $\sigma$-finite projections is again $\sigma$-finite.
We also recall that a von Neumann algebra is said to be \emph{$\sigma$-finite} if its identity is a $\sigma$-finite projection.


\begin{thm}\label{thm:Raggio-tran-prop<->jord-isom}
	Let $M_1$ and $M_2$ be two von Neumann algebras, and let $\Phi:\KS(M_1)\to \KS(M_2)$ be a bijection.
	Then $\Phi$ preserves the Raggio transition probability if and only if one can find a (necessarily unique) Jordan $^*$-isomorphism $\Theta:M_2\to M_1$ satisfying
	$\Phi = \Theta_*|_{\KS(M_1)}$.
\end{thm}
\begin{prf}
One direction of the equivalence follows from Proposition \ref{prop:J-preserves-U-R-TP}.
For the opposite direction, we assume in the following that $\Phi$ preserves the Raggio transition probability.

Notice that because of Relation \eqref{eqt:orth<->van-tran-prob}, the map $\Phi$ is biorthogonality preserving, and Lemma \ref{lem:st-sp-to-proj} gives an orthoisomorphism $\check \Phi: \Proj(M_1)\to \Proj(M_2)$.
Moreover, by Relations \eqref{eqt:Rag-tran-prob<->d2} as well as \eqref{eqt:d2<->Hil}, the map $\varphi: S_{\KP_1}\to S_{\KP_2}$ given by
\begin{equation*}\label{eqt:def-varphi}
\varphi(\xi_\mu):= \xi_{\Phi(\mu)} \qquad (\mu\in \KS(M_1)
\end{equation*}
is a metric preserving surjection, and Proposition \ref{prop:isom-cone>Jord} tells us that it extends to a complex linear isometry $\bar \varphi:\KH_1\to \KH_2$ satisfying $\bar \varphi(\KP_1) = \KP_2$.

By considering finite sums of elements in $\Proj_\sigma(M_1)$, one obtains, through \eqref{eqt:proj-decomp}, an increasing net
$\{e_i\}_{i\in \KI}$ of $\sigma$-finite projections such that $e_i\uparrow 1$.
Let us put $f_i:= \check \Phi(e_i)$ ($i\in \KI$).
Then all $f_i$ are $\sigma$-finite and $f_i\uparrow 1$ (because $\check \Phi$ is an orthoisomorphism).

By Corollary 2.5 and Lemma 2.6 of \cite{Haag-st-form}, the standard form for $e_i M_1 e_i$ is
$$
\big(e_i M_1 e_i, e_ie_i^\bt\KH_1, e_ie_i^\bt\KP_1, e_ie_i^\bt J_1 e_ie_i^\bt\big)
$$
(observe that $e_i^\bt xe_i^\bt \eta = xe_i^\bt \eta = x\eta$, whenever $x\in e_iM_1e_i, \eta\in e_ie_i^\bt \KH_1$).
In a similar way,
$\big(f_i M_2 f_i, f_if_i^\bt\KH_2, f_if_i^\bt\KP_2, f_if_i^\bt J_2 f_if_i^\bt\big)$ is the standard from of $f_i M_2 f_i$.

We identify, as usual,
$$\KS(e_iM_1e_i) \cong e_i \KS(M_1)e_i \cap \KS(M_1) = F_0(1-e_i)$$
and $\KS(f_iM_2 f_i) \cong F_0(1-f_i)$ in the canonical ways.
From this, the map $\Phi$ induces, through Lemma \ref{lem:st-sp-to-proj}(a), a bijection $\Phi_i: \KS(e_iM_1e_i) \to \KS(f_i M_2 f_i)$.
For each $\mu\in \KS(e_iM_1e_i)$, let $\xi^i_\mu\in S_{e_ie_i^\bt \KP_1}$ be the element with
$\mu(x) = \la x\xi^i_\mu, \xi^i_\mu\ra$ $(x\in e_i M_1 e_i)$.
Then
$$\mu(y) = \mu(e_i y e_i) = \la y\xi^i_\mu, \xi^i_\mu\ra \qquad (y\in M_1),$$
and the uniqueness of the element $\xi_\mu$ in $\KP_1$ satisfying \eqref{eqt:nor-st-P} implies that $\xi^i_\mu = \xi_\mu$.
Hence, if $\varphi_i: S_{e_ie_i^\bt \KP_1} \to S_{f_if_i^\bt \KP_2}$ is the bijection defined by $\varphi_i(\xi^i_\mu):= \xi^i_{\Phi_i(\mu)}$ ($\mu\in \KS(e_i M_1 e_i)$), we have $\varphi_i = \varphi|_{S_{e_ie_i^\bt \KP_1}}$.

The above shows that $\psi_i:=\bar\varphi|_{e_ie_i^\bt\KH_1}$ is a bijective isometry from $e_ie_i^\bt\KH_1$ to $f_if_i^\bt\KH_2$ with $\psi_i(e_ie_i^\bt\KP_1) = f_if_i^\bt\KP_2$.
Since both $e_iM_1 e_i$ and $f_iM_2 f_i$ are $\sigma$-finite, \cite[Th\'{e}or\`{e}me 3.3]{Connes-Cara} gives a Jordan $^*$-isomorphism $\Lambda_i:e_i M_1 e_i \to f_i M_2 f_i$ such that for every $x\in e_iM_1e_i$ and $\xi\in S_{e_ie_i^\bt\KP_1}$,
\begin{equation}\label{eqt:rel-Lambda-i-Phi-i}
\Phi(\omega_\xi)(\Lambda_i(x)) = \omega_{\varphi(\xi)}(\Lambda_i(x)) = \la \Lambda_i(x)\varphi_i(\xi), \varphi_i(\xi)\ra = \la x\xi, \xi\ra.
\end{equation}
In particular, one has $\Phi_i = (\Lambda_i^{-1})^*|_{\KS(e_i M_1 e_i)}$.

As in the beginning of the proof, $\Phi_i$ is biorthogonality preserving and induces an orthoisomorphism $\check \Phi_i: \Proj(e_i M_1 e_i)\to \Proj(f_i M_2 f_i)$ satisfying Relation \eqref{eqt:full-rel-Psi-F0}.
It then follows from
$$\Phi\big(F_0(p)\cap \KS(e_i M_1 e_i)\big) = F_0\big(\check \Phi(p)\big)\cap \KS(f_i M_2 f_i) \qquad (p\in \Proj(e_i M_1 e_i))$$
that $\check \Phi_i = \check \Phi|_{\Proj(e_i M_1 e_i)}$.
Thus, Lemma \ref{lem:st-sp-to-proj}(c) implies
$\Lambda_i|_{\Proj(e_i M_1 e_i)} = \check \Phi|_{\Proj(e_i M_1 e_i)}$.
From this, we know that whenever $i\leq j$, one has $\Lambda_j|_{\Proj(e_i M_1 e_i)} = \Lambda_i|_{\Proj(e_i M_1 e_i)}$, which ensures that
$$\Lambda_j|_{e_i M_1 e_i} = \Lambda_i.$$

Set $M_1^e:= \bigcup_{i\in \KI} e_i M_1 e_i$ and $M_2^f:= \bigcup_{i\in \KI} f_i M_2 f_i$.
The above allows us to define a Jordan $^*$-isomorphism $\Lambda_0:M_1^e \to M_2^f$ satisfying
$\Lambda_0|_{e_i M_1 e_i} = \Lambda_i$, and \eqref{eqt:rel-Lambda-i-Phi-i} gives
\begin{equation}\label{eqt:rel-Lambda-0-Phi}
\omega_{\varphi(\xi)}(\Lambda_0(x)) = \omega_\xi(x)
\qquad (x\in M_1^e, \xi\in \KP_1)
\end{equation}
because $\varphi$ is an isometry and $\bigcup_{i\in \KI} e_ie_i^\bt \KP_1$ is norm-dense in $\KP_1$.
We thus know from $\{\omega_{\varphi(\xi)}: \xi\in \KP_1\} = (M_2)_*^+$ that $\Lambda_0$ is weak$^*$-continuous.

On the other hand, since $e_iye_i \stackrel{w^*}{\longrightarrow} y$ for any $y\in M_1$ (see e.g.\ Remark \ref{rem:WOT=weak*}), $M_1^e$ is weak$^*$-dense in
$M_1$.
Hence, $\Lambda_0$ extends to a weak$^*$-continuous complex linear map $\Lambda:M_1 \to M_2$ such that $\Lambda(M_1^+)\subseteq M_2^+$, $\Lambda(1) = 1$ and, because of \eqref{eqt:rel-Lambda-0-Phi},
\begin{equation}\label{eqt:rel-Phi-Lambda}
\Phi(\mu)(\Lambda(x)) = \mu(x) \qquad (x\in M_1, \mu\in \KS(M_1)).
\end{equation}

Similarly, $\Phi^{-1}$ induces a positive linear map $\Upsilon: M_2\to M_1$ satisfying the corresponding property as \eqref{eqt:rel-Phi-Lambda}.
Clearly, $\Upsilon$ is the inverse of $\Lambda$, and $\Lambda$ is an order isomorphism.
By \cite[Corollary 5]{Kad52}, $\Lambda$ is a Jordan $^*$-isomorphism, and $\Theta := \Lambda^{-1}$ is the required map.
\end{prf}


The proof above can be shorten quite a bit if \cite[Th\'{e}or\`{e}me 3.3]{Connes-Cara} holds for the non-$\sigma$-finite case. 
However, this seems to be unknown. 
Note that even in the later work of \cite[Theorem VII.1.1]{Ioch-sd-cone}, which generalised \cite[Th\'{e}or\`{e}me 3.3]{Connes-Cara} to the case of $JBW^*$-algebras, the $\sigma$-finite assumption was still imposed.


The following corollary is a direct consequence of Theorems \ref{thm:st-sp-met-inv}(b) and \ref{thm:Raggio-tran-prop<->jord-isom} as well as Proposition \ref{prop:J-preserves-U-R-TP}.


\begin{cor}\label{cor:pre-Raggio>pre-Uhlmann}
If $\Phi:\KS(M_1)\to \KS(M_2)$ is a bijection preserving either the Raggio transition probability or the ``asymmetric transition probability'', then it preserves the Uhlmann transition probability as well.
\end{cor}

It is natural to ask if the converse of the above holds.
This lead to the following question.

\begin{question}\label{question:Hhlmann}
If $\Phi: \KS(M_1)\to \KS(M_2)$ is a bijection preserving the Uhlmann transition probability, can one find  a Jordan *-isomorphism
$\Theta: M_2\to M_1$ satisfying $\Phi = \Theta_*|_{\KS(M_1)}$?
\end{question}

Let us end this section with another application of our main results.
Here,  $d_{\|\cdot\|}$ denotes the metric on $\KP$ defined by the norm on $\KH$.


\begin{cor}\label{cor:comp-Jord-inv}
	Let $M$ be a von Neumann algebra.
	Each one of the following metric spaces: $(\KP, d_{\|\cdot\|})$, $(\KS(M), \dist_B)$, $(\KS(M), \dist_1)$ and $(\KS(M), \dist_2)$ is a complete Jordan $^*$-invariant for $M$.
\end{cor}
\begin{prf}
The fact that $(\KS(M), \dist_1)$ is a complete Jordan $^*$-invariant for $M$ is already proved in Theorem \ref{thm:st-sp-met-inv}.
Moreover, it follows from Theorem \ref{thm:st-sp-met-inv} and Relation \eqref{eqt:Rag-tran-prob<->d2} (respectively, \eqref{eqt:def-d-B}) that $(\KS(M), \dist_2)$ (respectively, $(\KS(M), \dist_B)$) is a complete Jordan $^*$-invariant.
Consequently, $(\KP, d_{\|\cdot\|})$ is also a complete Jordan $^*$-invariant because of Proposition \ref{prop:isom-cone>Jord}.
\end{prf}


Furthermore, if $\dim M \geq 2$, then the metric space $B_\KP:=\{\xi\in \KP: \|\xi\|\leq 1\}$ (under the metric induced by the norm on $\KH$) is also a complete Jordan $^*$-invariant for $M$.
In fact, it is not hard to see that the set of extreme points of $B_\KP$ is $S_\KP\cup \{0\}$.
Thus, using the fact that $S_{\KP_2}$ is not a singleton set, a continuity argument  will verify that if $\phi: B_{\KP_1}\to B_{\KP_2}$ is a distance preserving bijection, then $\phi(0) = 0$.
One can find the details of this, as well as its generalization to all non-commutative $L_p$-spaces ($p\in (1,\infty)$), in our further work on
this subject (\cite{LNW-nc-Lp}).

\section{Acknowledgement}

The authors are supported by National Natural Science Foundation of China (11471168) and Taiwan MOST grant (102-2115-M-110-002-MY2).


\bibliographystyle{plain}

\end{document}